\documentclass[10pt]{article}
\usepackage{ifpdf}
\usepackage{graphicx,natbib,amssymb,amsmath}
\usepackage{fancyhdr}
\usepackage[T1]{fontenc}
\usepackage{latexsym,theorem}
\usepackage{amsmath,amssymb}
\usepackage{graphicx}
\usepackage{epsfig}
\usepackage{url}
\usepackage{xcolor}
\usepackage{euscript}
\usepackage{fancyhdr}
\usepackage{ifpdf}

\def\mathscr{\EuScript}

\newcommand{\argmin}{\mathop{\arg \! \min}}                 % Arg-min
\usepackage{natbib,amssymb,amsmath}
\usepackage{euscript}
\usepackage[english]{babel}

\theoremstyle{plain}
\newtheorem{theorem}{Theorem}

\newtheorem{lemme}{Lemma}
\newtheorem{proposition}{Proposition}

{\theorembodyfont{\upshape}
\newtheorem{definition}{Definition}

\newenvironment{Proof}{\small{\bf Proof.}}{\hfill$\Box$\normalsize
\bigskip}

\usepackage{xcolor}

\newcommand{\RR}{{\mathbb R}} %
\newcommand{\EE}{{\mathbb E}} %
\newcommand{\PP}{{\mathbb P}} %

\newcommand{\Risk}{{\mathcal R}} %
\newcommand{\1}{{\mathbf 1}}

\def\cvar{C\!V\!a\!R}

\def\VaR{V\!a\!R}

\def\gain{G}

\def\Om{\Omega}
\def\F{\mathcal{F}}

\def\gain{G}

\def\text#1{\quad\mbox{#1}\quad} %wide text in mathsok
\def\defegal{\triangleq}

\def\gain{X}

\def\Om{\Omega}
\def\F{\mathcal{F}}

\newenvironment{proof}{\small{\bf
    Proof:}}{\hfill$\Box$\normalsize\\\bigskip}

\newenvironment{keywords}{\small{\bf
    Key Words:}}{\hfill\normalsize\\ \bigskip}

\title{Computational Dynamic  Market Risk Measures in Discrete Time Setting} 
\author{ %
Babacar Seck\footnote{
Department of Mathematics and Statistics,
University of Calgary
2500 University Drive NW
Calgary, AB, T2N 1N4, Canada.
Email: babacar@ucalgary.ca}\,,
Robert J.  Elliott\footnote{
Haskayne School of Business, 
University of Calgary Calgary, 
2500 University Drive NW
Calgary, AB, T2N 1N4, Canada.
Corresponding author:  relliott@ucalgary.ca
}\,, 
Jean-Pierre Gueyie\footnote{
D\'{e}partement de Finance,
Universit\'{e} du Quebec \`{a} Montr\'{e}al,
315, rue Sainte-Catherine Est
Local R-3555
Montr\'{e}al, Qu\'{e}bec, H2X 3X2, Canada.
Email: gueyie.jean-pierre@uqam.ca}
}

\begin{document}
\maketitle

\setcounter{tocdepth}{3} %LC-35

\vspace{0.5cm}
\begin{abstract}
Different approaches to defining dynamic market risk measures are available in the literature. Most  are focused or derived 
from probability theory, economic behavior or dynamic programming. Here, we propose an approach to define and implement 
dynamic market risk measures based on recursion and state economy representation. The proposed approach is  to be 
implementable and to inherit properties from static market risk measures.
\end{abstract}


\section{Introduction}

An understanding of  financial risk and related methods used to calculate risk indicators are necessary for a healthy economy, 
as evidenced by recent  economic crises. For commercial banks, regulations and rules  advocated by the Basel Committee I, II and  III are in terms of
minimum capital requirements, supervisory reviews and market discipline; see Basel~(\citeyear{Basel2,Basel3}).
Even if the methods used for the market risk are now well established, and the measure used is  
Value-at-Risk, the dynamic counterpart of this market risk is approximate. Recently, axioms have been  developed for  risk measures by
Delbaen et~al. in~\citep{Delbaen1999}. Proposals have been made to extend these properties to different model settings particularly  to a dynamic framework in~\cite{Delbaen2006,Delbaen2006bis,Delbaen2006three,Jouini2008,Stadje2010}, and references therein.
A recent review of the literature on dynamic market risk measures is available in~\cite{Acciaio2010}.
In the spirit of these studies, the problem considered here is how to build dynamic risk measures based on  static measurements which can easily be implemented and which satisfy certain properties.
We do so by making use of a recursion formula based on a given static risk measure $\Risk$ and  a state economy representation modeled by a discrete Markov chain. With these approaches, the dynamic risk measure inherits properties from the static risk $\Risk$.
The computational complexity of the dynamic risk measure is then reduced to the computational complexity of the static risk one. In term of interpretation,  the Markov chain captures the uncertainties in the economy while the decision maker's behavior is expressed via a recursion formula. 
 The paper is organized as follows. In section~\ref{properties}, usual properties of market risk measures and their interpretations are briefly recalled. 
Then section~\ref{newformul} presents the recursive principle and the state economy approach to formulate dynamic market risk. These formulations of dynamic risk measures inherit properties from static risk measures. In section~\ref{applic}, closed formulae for recursive risk measures and Markov modulated risk measures are given through the two main market risk measures:  Value-at-Risk and  Conditional Value-at-Risk.  An application to those two risk measures is shown when considering Gaussian and Weibull returns. The recursive Value-at-Risk presents closed formula easy to compute in the sense that at any time $t$, only the parameters of the distributions up to time $t$ are to be considered. The simulations for a ten day horizon time reveal a relatively small and stable Markov modulated risk measures which can contribute to lower the capital requirement. The conclusion is given in section~\ref{conclu}.

\section{Properties of risk measures}
\label{properties}
The usual properties of market risk measures are briefly presented, in both static and dynamic settings.

\subsection{Properties of  static risk measures}

Consider a function $\Risk$ defined  from a set of random variable  $\mathcal{X}$ into 
$\RR\cup\{-\infty\}\cup\{+\infty\}=\overline{\RR}$.  
 For the sake of simplicity,
suppose $\mathcal{X}$ is a set of real valued random variables such that

\begin{equation}
\forall\, X \in \mathcal{X},\,\forall\, m\in\RR,\quad X+m\in\mathcal{X}.
\end{equation}

\begin{definition}
 $\Risk$ is a (static) monetary risk measure if it satisfies the following two properties:
\begin{enumerate}
\item[$P_1$.] Monotonicity: $\forall\,X,\,Y\in
\mathcal{X},\,\,\, 
X \leq Y \Rightarrow\,\Risk(X)\geq \Risk(Y)$\, ,
\item[$P_2$.] Translation invariance: $\forall\,
m\in\RR,\,\,\,\Risk(X+m)=\Risk(X)-m$\,.
\end{enumerate}
\end{definition}
With monotonicity, if the position  $X$ is not more valuable than the position $Y$, then the risk associated with $X$ will be greater than the risk associated with $Y$. Translation invariance stipulates that additional cash decreases the risk, in particular
$\Risk(X+\Risk(X))=0$.

\begin{definition}
A risk measure $\Risk$ is said to be coherent if it satisfies the properties  $P_1$, $P_2$, and:
\begin{enumerate}
\item[$P_3$.] Sub-adititivity:
  $\forall\,X,Y\in\mathcal{X}$,\quad 
$\Risk(X+Y)\leq \Risk(X) + \Risk(Y)$\,,
\item[$P_4$.] Positive homogeneity: if $ k > 0$ and if 
$X\in\mathcal{X}$ then $\Risk(k X)=k\Risk(X)$\,.
\end{enumerate}
\end{definition}

With sub-additivity, the risk decreases by diversification.
Positive homogeneity implies that the risk is proportional to the size of the portfolio. 
Further discussions on  risk measure axiomatics can be found in~Artzner~\emph{et al.}~(\citeyear{Delbaen1999,Delbaen2002}).  

\subsection{Properties of  dynamic risk measures}
\label{dynamicmodel}

Fix a finite time horizon $T>0$. Contingent claims are represented by  random variables $X$ defined on a probability space
$(\Omega,\mathcal{F},\PP)$ with a filtration $(\mathcal{F}_{t})_{t\in S}$.
Here $S \subset [0;T]$ is a set of times  including $0$ and the time horizon $T$.
In the literature, two different approaches are considered to describe dynamic  market risk: either  $X$ gives a final payoff at time $T$, 
 or $X$ defines a stochastic process which takes into account payments  at times $t<T$, denoted $X_t$.
 In both cases, the problem of defining dynamic risk measures that are time consistent is crucial. Time consistency expresses the way that the risk measurements at different time are interrelated. Different notions of time consistency are introduced in the literature. Here, by time consistent we mean strongly time consistent. 
 
 Consider a family of maps $\Risk_{t}\!\!:\;L^{\mathrm{p}}(\Om,\F_{T},\PP) \to L^{\mathrm{p}}(\Om,\F_{t},\PP)$. %  a dynamic risk measure.
 To present the properties of dynamic risk measures, assume that $X$ gives a final payoff at time $T$ and is defined on 
 $L^{\mathrm{p}}(\Omega,\mathcal{F}_{T},\PP)$. % with value on $(\RR,\mathcal{B})$ where $\mathcal{B}$ stands for the Borel $\sigma$-field.
Denote  by $\Risk_{t}(X)$ the risk  associated with $X$ at time $t$.
$\Risk_{t}(X)$ is usually interpreted as a minimum capital requirement at time $t$ knowing the information $\F_{t}$.

The properties of a dynamic risk measure are:
\begin{itemize}
\item[$D_1$.] Normalization: $\Risk_{t}(0) = 0$, 
\item[$D_2$.] Monotonicity: $\forall\, X,\,Y\in L^{\mathrm{p}}(\Om,\F_{T},\PP),\quad X \leq  Y \Rightarrow \Risk_{t}(X)\geq \Risk_{t}(Y)$,
\item[$D_3$.] Translation invariance:  $\forall X\in 
L^{\mathrm{p}}(\Om,\F_{T},\PP),\,\forall m\in L^{\infty}(\Om,\F_t,\PP) , 
\quad \Risk_{t}(\gain + m )=
\Risk_{t}(\gain) - m $,
\item [$D_4$.] Local property: $\forall X,Y\in L^{\mathrm{p}}(\Om,\F_{T},\PP) ,\; \forall A\in \F_{t},
\quad \Risk_{t}(1_{A}X+1_{A^{c}}Y)=1_{A} \Risk_{t}(X)+ 1_{A^{c}}\Risk_{t}(Y) $,
\item [$D_5$.] Time consistency:  $\forall X,Y\in L^{\mathrm{p}}(\Om,\F_{T},\PP),\quad \Risk_{t}(X)\leq \Risk_{t}(Y)\Rightarrow
\Risk_{s}(X)\leq \Risk_{s}(Y), \; \forall s,t\in S$ with $s\leq t$.
\end{itemize}
After normalization, a null position does not require any capital reserve to make the risk zero. 
 The local property implies that,
if the event $A$ is $\F_{t}$-measurable then the decision maker should know at  time $t$ if $A$ has happened and adjusts his evaluation accordingly. 
 In the literature many definitions of time consistency appear, resulting from preferences and decision policy interpretations. Time consistency is derived 
 from the so-called Bellman principle;  see~\citep{Bellman1962,Bertsekas1996_2}.

A dynamic risk measure is convex if it satisfies $D_1$, $D_2$, $D_3$ and the following property:

\noindent \mbox{ $D_6$. \emph{Convexity}: }
 $\forall X,\,Y\in  L^{\infty}(\Om,\F_{T},\PP)$ and any
$\eta:  L^{\infty}_+(\Om,\F_t,\PP)\to [0,1]$
\begin{equation}
 \Risk_{t} (\eta X +
(1-\eta) Y)\leq \eta\Risk_{t}(X) +
 (1-\eta) \Risk_{t}(Y)\,,
\end{equation}
where $L^{\infty}(\Om,\F_{t},\PP)$ is the space of bounded $\F_t$-measurable random variables and $L_{+}^{\infty}(\Om,\F_{t},\PP)\defegal \big\{ \gain\in   L^{\infty}(\Om,\F_{t},\PP)\mid \PP\big(\gain \geq 0\big)=1\big\} $\,.
The convexity property says that diversification should decrease the risk.

 The coherence property is satisfied if the maps $\Risk_t$
satisfies $D_1$ to $D_5$ and the positive homogeneity property.

\noindent \mbox{$D_7$. \emph{Positive homogeneity: } }
\begin{equation}
\forall\,X\in L^{\mathrm{p}}(\Om,\F_{T},\PP),\,
 \forall\,m \in L_{+}^{\infty}(\Om,\F_{t},\PP)\, \quad 
\Risk_{t}(m \gain)= m\Risk_{t}(\gain)\,.
\end{equation}

Usually, static risk measures are associated with corresponding acceptable positions. The acceptable sets describe
the set of positions which are immunized against the uncertainties.
The following definition extends this notion to the dynamic case.
 
\begin{definition}
Let $\Risk_{t}: L^{\mathrm{p}}(\Om,\F_{T},\PP) \to L^{\mathrm{p}}(\Om,\F_{t},\PP)$ be a dynamic risk measure. We define the acceptable set
associated with $\Risk_{t}$ by:
\begin{equation}
\mathcal{C}_{\Risk_t}:=\left\{\gain\in L^{\mathrm{p}}(\Om,\F_{T},\PP) \mid
\Risk_t(\gain)\leq 0 
\right\}\,.
\end{equation}
\end{definition}
The properties of a dynamic risk measure $\Risk_t$ can be expressed via the associated acceptable set
 $ \mathcal{C}_{\Risk_t}$ and vice versa.

 \begin{proposition}
Suppose that  $\Risk_{t}$ satisfies the properties $D_1$, $D_2$ and $D_3$. 
Then, $\mathcal{C}_{\Risk_t}$ satisfies the following properties.
\begin{itemize}
\item $\mathcal{C}_{\Risk_{t}}$ is not empty and satisfies the following properties:
\begin{subequations}
\begin{equation}
\inf\{m\in L^{\mathrm{p}}(\Om,\F_{T},\PP)\mid m\in \mathcal{C}_{\Risk_{t}}  \}> -\infty\,,
\end{equation}
\begin{equation}
\forall X\in \mathcal{C}_{\Risk_{t}},\,\forall\, Y\in L^{\mathrm{p}}(\Om,\F_{T},\PP),\quad
 Y \leq X \Rightarrow Y\in \mathcal{C}_{\Risk_{t}}\,.
\end{equation}
\label{eq:simples}
\end{subequations}
\item Let $\mathcal{C}_{\Risk_{t}}$ be given. The associated dynamic risk measure $\Risk_{t}$ 
is:
\begin{equation}
\Risk_{t}( X )=ess\inf\{m\in L^{\mathrm{p}}(\Om,\F_{T},\PP)\mid m+X\in\mathcal{C}_{\Risk_{t}} \}\,.
\label{eq:risk_measure-acceptable}
\end{equation}
\item $\Risk_{t}$ is convex if and only if the associated acceptable set
$\mathcal{C}_{\Risk_{t}} $ is   conditionally convex:
\begin{equation}
\forall V,\,W\in \mathcal{C}_{\Risk_{t}},\quad \beta V + (1-\beta)W \in \mathcal{C}_{\Risk_{t}},\;\mbox{with} \;\beta\!\!: L^{\infty}_+(\Om,\F_t,\PP)\to [0,1].
\end{equation}

\item  $\Risk_{t}$ is coherent if and only if the associated acceptable set 
$\mathcal{C}_{\Risk_{t}} $ is a conditionally  convex cone.
\end{itemize}
\end{proposition}

\begin{Proof} See~\citep{Delbaen2006,Roorda2005}.
\end{Proof}

In the rest of this paper, we assume that $X=(X_t)_{t\geq 0}$ is a finite dimensional stochastic process defined on the filtered probability space
$(\Om,\F,(\F_t)_{t\geq 0},\PP)$.% and taking values in $(\RR,\mathcal{B})$.

\section{New formulation of dynamic risk measures}
\label{newformul}
A recursive principle and a state economy approach to formulate dynamic market risk are introduced. The new formulations of dynamic risk measures are relatively easy to compute and inherit properties from static risk measures.

\subsection{Dynamic risk measures based on a recursion formula}

We now construct a dynamic risk measure, denoted by $(\Risk_t)_{t\in S} $ which is
based on a static risk measure $\Risk\!:\mathcal{X} \to \overline{\RR}$. % for $t=0,\ldots,T$. 
Suppose that the time index set is discrete and finite: $S\defegal \{0,\ldots,T \}$.

\begin{definition} 
\label{defdynamic1}
Suppose that a static risk measure  $\Risk$ is given. 
We define a dynamic recursive risk measure by the collection  of functions $\Risk_t$ given by:
\begin{equation}
\label{dynamic1}
\forall F \in\F_t,\quad  \Risk_{t}(X)\defegal \1_{F}\Risk (X_t+\Risk_{t-1}(X) ),\quad \Risk_0(X)\defegal \Risk(X_0),
\end{equation} 
where $\forall \omega\in\Om,\;\1_{F}(w)=1$ if $F\in\F_t$ and $\1_{F}(w)=0$ elsewhere.
\end{definition}

The interpretation of the above definition is the following: at time $t=0$, the decision maker knows exactly the amount at risk $\Risk(X_{0})$. 
At time $t=1$, to immunize the position $X_{1}$ against  uncertainties, based on his knowledge at time $t=0$, the capital 
$\Risk(X_{0})$ is required. 
Then, the risky position at time $t=1$ becomes
$X_{1}+\Risk(X_{0})$.  So, the  capital at risk at time $t=1$ based on the risk measure $\Risk$  
and the information generated by the stochastic process $X$  
 is $\1_{F}\Risk( X_1+\Risk(X_{0}))$. %,\; \forall A\in\F_\1$. 
The risk at any time $t$ is then built with  recursion knowing the risk at time $t-1$. The above recursion formula of dynamic risk measures is similar to the recursive utility functions developed 
by~\cite{Epstein1989,Epstein1991}.
The principle  introduced here differs from the recursion utility, where the recursion is based on information available in a planning period.

 The collection of maps $\Risk_{t}$  
  inherits properties from the static risk measure $\Risk$. 
\begin{proposition}
\label{propositiondynamic}
Suppose that $X$ and $Y$ are two stochastic processes.
\begin{enumerate}
\item Suppose that the static risk measure $\Risk\!\!: \mathcal{X} \to \overline{\RR}$ satisfies the monotonicity property
 and $\Risk(X_{0})\leq \Risk(Y_{0})$. 
Then the dynamic recursive risk measure $\Risk_{t}$  is monotone. 

\item Suppose that the static risk measure $\Risk\!\!: \mathcal{X} \to \overline{\RR}$ 
satisfies the translation invariance property and the function $\Risk$ is additive.
Then, the dynamic recursive risk measure $\Risk_{t}$ is invariant by translation.
\item Suppose that $\Risk\!\!: \mathcal{X} \to \overline{\RR}$ is positively homogeneous and additive.
Then, the associated dynamic recursive risk measure  $\Risk_{t}$ satisfies the local property.
\item Suppose that $\Risk\!\!: \mathcal{X} \to \overline{ \RR}$ is convex and monotone.
Then, the associated dynamic recursive risk measure  $\Risk_{t}$ is convex.
\end{enumerate}
\end{proposition} 
We shall demonstrate $1.$;  the points $2.$,  $3.$ and $4.$ are straightforward.\\

\begin{proof}
\begin{enumerate}
\item Let $X$ and $Y$ be two stochastic processes such that $X \leq Y,\; \PP\, a.s$. It follows that for almost every $\omega\in \Omega$, 
\begin{equation*}
X_{t}(\omega) \leq Y_{t}(\omega),\quad t=1,\dots,T.
\end{equation*}
In particular,  $X_{1}(\omega) +\Risk(X_{0})  \leq Y_{1}(\omega)  + \Risk(Y_{0})  $. By monotonicity of $\Risk$ it follows that
\begin{equation*}
\forall F \in\F_1,\quad \1_{F}\Risk(X_{1} +\Risk(X_{0}))  \geq \1_{F}\Risk( Y_{1}  + \Risk(Y_{0} )).
\label{recur1}
\end{equation*}
Then $\Risk_{1}(X)\geq \Risk_{1}(Y)$. 

By recursion, 
$$ X_{t}(\omega) \leq Y_{t}(\omega) \Rightarrow \Risk_{t}(X)\geq \Risk_{t}(Y)\quad \mbox{for}\quad t=1,\ldots,T.$$
Consequently $\Risk_{t}$ is monotone.

\end{enumerate}
\end{proof}

\begin{definition}
Suppose that a dynamic recursive risk measure $\Risk_{t}$ is given.
The corresponding acceptable sets can be defined as: 
\begin{equation}
\mathcal{C}_{\Risk_{t}}\defegal \{  X\in L^{\mathrm{p}}(\Om,\F_T,\PP) \mid  \Risk_{t}(X)\leq 0 \}.
\end{equation}
\end{definition}
It is straightforward to see that  the above acceptable set can be derived from the acceptable set of the static risk measure.  Then, the acceptable sets associated with the dynamic recursive risk measure  $\Risk_{t}$ can be expressed via the  acceptable 
set of the static risk $\Risk$ and inherits its properties, as  in Proposition~\ref{propositiondynamic}.

\subsection{Markov modulated dynamic risk measures}

The aim of this section is to define dynamic risk measures which take into consideration the state of the economy
for a given  risk measure.
Let  $Z=\{ Z_{t}, t \geq 0 \} $ be a discrete-time, finite state Markov chain with state space $\mathcal{Z}\defegal
\{z_{1},\ldots,z_{N}  \}  $ defined on the complete probability space  $(\Omega,\F,\PP)$.
 The states of $Z$ are interpreted as different states of an economy.  Following~\citep{Elliott1994}, we shall represent the state space of $\mathcal{Z}$ as the set of unit vectors 
  $\mathcal{E}\defegal \{e_{1},\ldots,e_{N}  \}$ where  $e_{i}=(0,\ldots,0,1,0,\ldots,0)^{\top}\in\RR^{N} $ and
  $\top$ represents the transpose of a matrix or a vector.  $Z$ then has the following semi-martingale decomposition:
 \begin{equation}
 \label{markovchaine}
Z_{t}= A Z_{t-1} +M_{t}\in\RR^{N}\,,
 \end{equation}
where $A=[a_{ji}]_{i,j=1,\ldots,N}$ and $a_{ji}=\PP\big( Z_{t}=e_{j}|Z_{t-1}=e_{i}\big)$.
$A$ is the transition matrix of the Markov chain and $M_{t}$ is a martingale increment
 with respect to the filtration $\F^{Z}$ generated by $Z$.
 
  Assume that the stochastic processes $X$ and $Z$ are not independent. For instance, 
 assume that, at any time $t$, there exists an $N$-vectorial real valued random variable $\widetilde{X}_t$ which stands for $N$ possible random variables such that
\begin{equation}
X_t=<\widetilde{X}_t ,Z_{t+1} >,\quad t=0,\ldots,T.
\end{equation}

\begin{definition}
Suppose that the time index set is discrete and finite: $S\defegal \{0,\ldots,T \}$.
Let $\Risk_{t}\!\!: L^{\mathrm{p}}(\Om,\F_{T},\PP) \to L^{\mathrm{p}}(\Om,\F_{t},\PP)$ be a dynamic risk measure. We define the 
state based dynamic risk measure $\Risk_{t}^{Z}$ as
\begin{equation}
\Risk_{t}^{Z}(X) \defegal \EE_{\PP}\big[ \Risk_{t}(X)| \mathcal{F}_{t}\vee \mathcal{F}_{t}^Z   \big],\quad t=0,\ldots,T,
\end{equation}
where $(\mathcal{F}_{t}^Z)_{t=0,\ldots,T}$ is the filtration generated by the Markov Chain. 
\end{definition}

\subsubsection{Examples of Markov modulated dynamic risk measures}

 Suppose that the function $F$ is a real-valued function of the Markov chain $Z$. Then, $F$
 has a linear representation $F(Z_t)={\bf F}^{\top}Z_t$ where  ${\bf F}=({\bf F}_1,\ldots,{\bf F}_N)^{\top}$ and  ${\bf F}_i=F(e_i)$.
 The function $F$ can be interpreted as an aggregated value of the different states of the economy.  Different Markov modulated dynamic risk measures can be considered with respect to the definition of the underlying risk measure.

 \begin{definition} Let a  static risk measure $\Risk$ be given.
 Suppose that the dynamic risk measure is defined by:
\begin{equation}
\label{case112}
\Risk_{t}(X)=\Risk(X_{t})<{\bf F},Z_{t+1}>, \; t=0,\ldots,T, \footnote{We assume that $Z_{T+1}\defegal Z_{T}$ as the horizon time is $T$.}
\end{equation} 
where $<\cdot,\cdot>$ states for the scalar product in $\RR^N$.
 Then, the derived state based dynamic risk measure 
$\Risk^{Z}_{t}(X) \defegal \EE_{\PP}\big[ \Risk_{t}(X)|\mathcal{F}_t \vee \mathcal{F}_{t}^Z\big]$ is
\begin{equation}
\Risk^{Z}_{t}(X)=\Risk(X_t)<{\bf F},AZ_{t}>, \; t=1,\ldots,T,
\end{equation} 
 \end{definition}
 If  $\Risk_{t}(X)=\Risk(X_{t})<{\bf F},Z_{t+1}>$  then    
 
\begin{eqnarray}
\Risk^{Z}_{t}(X)&=& \EE\big[  <\Risk(X_{t}){\bf F},Z_{t+1}> | \mathcal{F}_t \vee \mathcal{F}_{t}^Z \big]\nonumber\\
&=& \Risk(X_t)<{\bf F},AZ_{t}> \nonumber. 
% \EE\Big[  \EE\big[  <\Risk(X_{t}){\bf F}^{\top},Z_{t+1}> | Z_{t}  \big] |Z_{t}  \Big]\nonumber. 
\label{eqrisk1}
\end{eqnarray}

If the stochastic processes $X$ and $Z$ are independent, we can still construct a Markov modulated dynamic risk measure  which depends on $Z$.  Let assume that any state of the economy can be associated to a particular risk measure. A possible interpretation of this statement is that the risk aversion of the decision maker depends on the state of economy. Hence, we shall consider, $N$-dimensional risk measures.  Let an  $N$-vectorial static risk measure $\overline{\Risk}$ be given, i.e. $\overline{\Risk}\defegal \big(\Risk^1,\ldots,\Risk^N\big)$ with
$\Risk^i\!\!:\mathcal{X}\to \overline{\RR} $, for $i=1,\ldots,N$.  Then, different Markov modulated dynamic risk measure can be derived.

\begin{definition} 
 Assume that the dynamic risk measure is given by:
\begin{equation}
\label{case1}
\Risk_{t}(X)=<\overline{\Risk}(X_{t})^{\top},Z_{t+1}>, \; t=0,\ldots,T. 
%\footnote{We assume that $Z_{T+1}\defegal Z_{T}$ as the horizon time is $T$.}
\end{equation} 
 Then, the derived Markov modulated dynamic risk measure 
$\Risk^{Z}_{t}(X) \defegal \EE_{\PP}\big[ \Risk_{t}(X)|\mathcal{F}_t \vee \mathcal{F}_{t}^Z\big]$ is
\begin{equation}
\Risk^{Z}_{t}(X)=<\overline{\Risk}(X_t),AZ_{t}>, \; t=1,\ldots,T,
\end{equation} 
where %$\phi_{j}(t)=<A{{\bf \bullet} j},\overline{\Risk}(X_t)>$,\quad $\phi(t)\defegal (\phi_{1}(t),\ldots,\phi_{N}(t))^{\top}$
  $\Risk^{Z}_{0}(X)\defegal \Risk (X_{0})$. 
\end{definition}

\begin{definition} 
 If the dynamic risk measure $\widetilde{\Risk}_{t}$ is given by:
\begin{equation}
\label{case2}
\widetilde{\Risk}_{t}(X)=<\overline{\Risk}_{t}(X)^{\top},Z_{t+1}>, \; t=0,\ldots,T, \; \mbox{where}\quad  
\overline{\Risk}_{t}(X)=\big(\Risk^{1}_{t},\ldots, \Risk^{N}_{t} \big), %\; \mbox{and}\\ \quad \Risk^{i}_{t}(X)=\Risk^{i}(X_t +\Risk_{t-1}^{i}(X)).
\end{equation} 
and $\Risk^{i}_{t}(X)=\Risk^{i}(X_t +\Risk_{t-1}^{i}(X))$, 
then the derived Markov modulated dynamic risk measure\\ $\widetilde{\Risk}^{Z}_{t}(X)=
\EE_{\PP}\big[ \widetilde{\Risk}_{t}(X)|\mathcal{F}_t \vee \mathcal{F}_{t}^Z\big]$ is
\begin{equation}
\widetilde{\Risk}^{Z}_{t}(X)=<\overline{\Risk}_{t}(X),A Z_{t}>, \; t=1,\ldots,T,
\end{equation} 
where %$\Phi_{j}(t)=<A_{{\bf \bullet j}},\overline{\Risk}_{t}(X)^{\top}>$,\quad $\Phi(t)\defegal (\Phi_{1}(t),\ldots,\Phi_{N}(t))^{\top}$ 
 $\widetilde{\Risk}^{Z}_{0}(X)\defegal \Risk (X_{0})$.
%\end{enumerate}
\end{definition}

\subsubsection{Properties of Markov modulated dynamic risk measures}
The Markov modulated dynamic risk measure $\Risk_{t}^{Z}$ inherits properties from the associated dynamic risk measure $\Risk_{t}$.
 \begin{proposition}
 \begin{enumerate}
 \item The state based dynamic risk measure $\Risk_{t}^{Z}$ is monotone  if the associated
 dynamic risk measure  $\Risk_{t}$ is monotone.
 \item The state based dynamic risk measure $\Risk_{t}^{Z}$ satisfies the local property if and only if
   $\Risk_{t}$ satisfies the local property on $\mathcal{F}_{t}\vee \mathcal{F}_t^{Z}$, where $\mathcal{F}_t^{Z}$
   stands for the filtration generated by the Markov chain $Z$.  
    \item The state based dynamic risk measure $\Risk_{t}^{Z}$ is invariant by translation if and only if
   $\Risk_{t}$ is invariant by translation on $\mathcal{F}_{t}\vee \mathcal{F}_t^{Z}$.
   \item  The state based dynamic risk measure $\Risk_{t}^{Z}$ is time consistent if 
  the associated dynamic risk measure $\Risk_{t}$ is time consistent.
 \end{enumerate}
 \end{proposition}

 \begin{proof} 
 We shall demonstrate the time consistent property. The other properties result  immediately from the conditional expectation properties.
Suppose that the dynamic risk measure associated with the Markov modulated dynamic risk,  $\Risk_{t}$ is time consistent.

Suppose that $\forall\,X,Y\in L^{\mathrm{p}}(\Om,\F_{T},\PP),\; \Risk_{t}(X)\leq \Risk_{t}(Y)$. By monotonicity of the conditional expectation operator, it follows
 
 $$\EE_{\PP}\big[\Risk_{t}(X)|\mathcal{F}_t \vee \mathcal{F}_{t}^Z\big] \leq  \EE_{\PP}\big[\Risk_{t}(Y)|\mathcal{F}_t \vee \mathcal{F}_{t}^Z\big]  \Longleftrightarrow \Risk^{Z}_{t}(X)\leq \Risk^{Z}_{t}(Y), \; \PP.\,a.s. $$
 
  By time consistency of $ (\Risk_{t})_{t=1,\ldots,T}$, we have
$$  \Risk_{t}(X)\leq \Risk_{t}(Y)   \Rightarrow     \Risk_{t-1}(X)\leq \Risk_{t-1}(Y), \; \PP.\,a.s.$$ It follows that 
$ \EE_{\PP}\big[\Risk_{t-1}(X)|\mathcal{F}_{t-1} \vee \mathcal{F}_{t-1}^Z\big] \leq 
 \EE_{\PP}\big[\Risk_{t-1}(Y)|\mathcal{F}_{t-1} \vee \mathcal{F}_{t-1}^Z\big] $. Then
$\Risk^{Z}_{t-1}(X)\leq \Risk^{Z}_{t-1}(Y)$. Hence, $\Risk_{t}^{Z}$ is time consistent.
 \end{proof}

\section{Application to Value-at-Risk and Conditional Value-at-Risk}
\label{applic}

Closed formulas of recursive risk measures and Markov modulated risk measures of the Value-at-Risk and the Conditional Value-at-Risk
are derived from the lemmas obtained in the previous section.  An application is then given when considering Gaussian and Weibull returns. The recursive Value-at-Risk gives closed formulas easy to compute in the sense, that at any time $t$, only the parameters of the distributions up to time $t$ are to be considered. 
The Markov modulated risk measures appears to be relatively small and stable,  which showed a benefit in term of capital management. 

\subsection{Value-at-Risk}
Value-at-Risk, $\VaR$, is the most widespread market risk measure used by the financial institutions.
In $1995$, the leading 10 banks in the OECD (Basel  committee) recommended this risk measure; the Basel II committee
advocated $\VaR$ as a standard risk measure.  
\begin{definition}
The Value-at-Risk associated with the random variable $X$ with value in $\RR$
at a level $p\in ]0,1[$ is defined as the smallest $p$-quantile of $X$:
\begin{equation}
\label{eq:var}
\VaR^{p}(X)\defegal
  \inf\left\{\eta\in\RR : \psi_X(\eta)\geq p\right\},
\end{equation}
where $\psi_X(\eta)\defegal \PP (X\leq\eta)$ defines the cumulative distribution function of $X$.
\end{definition}
The above minimum is reached because the function
$\psi_X$ is continuous on the left and non-decreasing. 
If $\psi_{X}$ is continuous and 
 increasing, then $\eta=\VaR^{p}(X)$
is the only solution of the equation  $\psi_{X}(\eta)=p$. 
Otherwise, this latter equation could have an infinity of solutions, (when the density
of $X$ is null for some value) or no solution at all, (when the density of probability is discrete).

For example, to say that the $\VaR^{99\%}$ of a portfolio  is  equal to 
 $\$100$ means that the maximal loss of value of the portfolio is less than a $\$100$ 
 with a probability of $99\%$.
Unfortunately, as a quantile, the $\VaR$ does not take into account extreme events and it is not sub-additive.
Diversification does not necessarily decrease  Value-at-Risk. Consequently,  $\VaR$ is not coherent or convex\footnote{The Value-at-Risk is used as an standard for  market risk but when using $\VaR$  diversification 
can fail to decrease the risk, as mentioned.}.
We  extend the static Value-at-Risk defined in~\eqref{eq:var} to the dynamic framework introduced
 in Section~\ref{dynamicmodel}. %~\eqref{dynamic1}. 
  In practice, three methods are used to compute $\VaR$: the historical method, the parametric method and 
 the Monte-Carlo method.  Each of them presents advantages and drawbacks. When considering the parametric method, the static 
 $\VaR$ can be extend to a dynamic framework under the hypothesis of normal distributions for the underling factors by  time series 
 estimation methods.

 In fact, if at time $t$, the random variable $X_{t} \sim \mathcal{N}(\mu_{t},\sigma_{t})$ then
 \begin{equation}
 \label{varnormal_t}
 \VaR^{p}(X_{t})\defegal \mu_{t} +\sigma_{t} q_{1-p},
 \end{equation}
where $q_{1-p}$  stands for the quantile of the standard normal distribution. Then, the value-at-risk associated with $X_{t+1}$ can be computed
once an  estimation method is set to determine $\mu_{t+1}$ and $\sigma_{t+1}$ knowing $\mu_{t}$ and $\sigma_{t}$ respectively. 
The case where the random variable $X_{t}$ follows a Weibull distribution will also  be discussed.  The Weibull distribution with shape parameter less than $1$
is a common heavy tail distribution which can  approximate many distribution functions. In practice, given a set of data,
it is often possible to find a Weibull distribution that fits the data  when the parameters of the Weibull distribution are calibrated, based on the variance and the mean of the
data. Recall that,  if $X_{t}$ is  a generalized Weibull distribution, with parameters $\lambda_{t}>0$, $\alpha_{t}>0$ and $\theta_{t}$,  
i.e. $X_{t}\sim\mathcal{W}(\lambda_{t},\alpha_{t};\theta_{t})$,
the associated cumulative distribution function $\psi_{X_{t}}$, is given by: (for $\eta>0$)

\begin{equation}
\label{eq:cumwei}
\psi_{X_{t}}(\eta)\defegal
  \left\{
\begin{array}{rl}
 1-e^{-\big(\frac{\eta-\theta_{t}}{\lambda_{t}}\big)^{\alpha_{t}}}, &                            \quad \mbox{if} \quad \eta \geq \theta_{t} , \\[2mm]
0, &      \quad\mbox{elsewhere}.
\end{array}
\right.
\end{equation}
$ \VaR^{p}(X_{t})$, the value-at-risk associated with $X_{t}$ is defined by $p=\psi_{X_{t}}\big( \VaR^{p}(X_{t}\big)$. Then
\begin{equation}
 \VaR^{p}(X_{t})=\displaystyle \theta_{t} + \lambda_{t}\Big(-\ln (1-p)\Big)^{\frac{1}{\alpha_{t}}}.
\end{equation}

Define $X\defegal (X_{0},X_{1},\ldots,X_{T})$.
\begin{lemme}
\label{lemme:gauss-weib}
Suppose that the random return $(X_{1},\ldots,X_{T})$ are independent and identically distributed.
\begin{enumerate}
 \item Suppose that $X_{t} \sim \mathcal{N}(\mu_{t},\sigma_{t})$
and $\VaR^{p}_{0}(X)=\VaR^{p}(X_{0})$.
The dynamic value-at risk based on recursion associated with $X$ is:
\begin{equation}
\VaR^{p}_{t}(X)=\displaystyle\sum_{k=1}^{t} (-1)^{t-k}\,(\mu_k+\sigma_{k}q_{1-p}) + (-1)^{t}\,(\mu_0+\sigma_{0}q_{1-p}).
\end{equation}

\item Suppose that $X_{t}\sim\mathcal{W}(\lambda_{t},\alpha_{t};\theta_{t})$  such that the cumulative distribution function  associated with  $X_{t}$
is defined in~\eqref{eq:cumwei}, 
and $\VaR^{p}_{0}(X)=\VaR^{p}(X_{0})$.
The dynamic value-at risk based on recursion associated with $X$ is:
\begin{equation}
\label{dy_var_recur_w}
\VaR^{p}_{t}(X)=
\displaystyle \sum_{k=1}^{t}(-1)^{t-k}\,\Big(\theta_{k}+\lambda_{k}\big(-\ln (1-p)\big)^{\frac{1}{\alpha_{k}}}\Big) 
+ (-1)^{t} \Big(  \theta_{0}+\lambda_{0}\Big(-\ln (1-p)\Big)^{\frac{1}{\alpha_{0}}} \Big).
\end{equation}

\end{enumerate}

\end{lemme}

\begin{proof} 
 For $t=1$, 
$$\VaR^{p}_{1}(X)=\VaR^{p}(X_{1}+\VaR^{p}_{0}(X))=\VaR^{p}\big(X_{1}+\VaR^{p}(X_{0})\big).$$
Using the invariance by translation property, it follows that 
$$\VaR^{p}_{1}(X)=\VaR^{p}(X_{1})-\VaR^p(X_0) .$$
By induction, for any given time $t\geq 1$ and for any natural number $n\leq t$, $\VaR^{p}_{t}(X)$ is given by:
\begin{equation}
\VaR^{p}_{t}(X)=\displaystyle \sum_{k=t-n+1}^{t} (-1)^{t-k}\,\VaR^{p}(X_k) + (-1)^{n}\,\VaR^{p}_{t-n}(X).
\end{equation}
 For $t=n$, we end up with the closed formula
 \begin{equation}
 \label{eq:closedformula}
\VaR^{p}_{t}(X)=\displaystyle \sum_{k=1}^{t} (-1)^{t-k}\,\VaR^{p}(X_k) + (-1)^{t}\,\VaR^{p}_{0}(X).
\end{equation}
%+ \mu_{0} +\sigma_{0} q_{1-p}\big).$$ 
\begin{enumerate}
 \item If  $X_{k} \sim \mathcal{N}(\mu_{k},\sigma_{k})$, then  $\VaR^{p}(X_k)= \mu_{k}+\sigma_{k}q_{1-p}$. 
 By definition, $\VaR^{p}_{0}(X)=\VaR^p(X_0)$.
 By substituting $\VaR^{p}(X_k)$ by its value in~\eqref{eq:closedformula}, we obtain the desired result.
%By induction, the desired result is obtained, for any discrete time $t >0$.
\item If $X_{k}\sim\mathcal{W}(\lambda_{k},\alpha_{k};\theta_{k})$, then $\VaR^{p}(X_k)= \theta_{k}+\lambda_{k}\big(-\ln (1-p)\big)^{\frac{1}{\alpha_{k}}} $. Hence by substituting $\VaR^{p}(X_k)$ by its value in~\eqref{eq:closedformula}, 
the equality~\eqref{dy_var_recur_w} is  obtained.

\end{enumerate}
\end{proof}

The dependence between the random returns $X_t$ and the state of the economy represented by the Markov chain can defined through 
the parameters of the density functions of the random returns, if it exists. For instance, 
\begin{enumerate}
\item if $X_{t} \sim \mathcal{N}(\mu_{t},\sigma_{t})$,  then $\mu_t\defegal <\mu,Z_{t+1}>\,\mbox{and}\,\sigma_t\defegal <\sigma,Z_{t+1}>$  with $\mu,\,\sigma\in\RR^N$; 
\item if $X_{t}\sim\mathcal{W}(\lambda_{t},\alpha_{t};\theta_{t})$, then $\lambda_t\defegal <\lambda,Z_{t+1}>,\;\alpha_t\defegal <\alpha,Z_{t+1}>$
and $ \theta_t\defegal <\theta,Z_{t+1}>$ with $\lambda,\,\alpha,\,\theta\in\RR^N$.
\end{enumerate}

\begin{lemme} 
\label{lemme:dynamicvar}
Suppose that  the returns $X_{t}$ are Gaussian and  independent where 
$\mu_t\defegal <\mu,Z_{t+1}>$ and $\sigma_t\defegal <\sigma,Z_{t+1}>$  with $\mu,\,\sigma\in\RR^N$.
Suppose the dynamic value-at-risk is driven by the recursion formula.
 Then the Markov modulated dynamic value-at-risk  is 
 \begin{equation}
  \Big(\VaR^{p}_{0}\Big)^{Z}(X)\defegal \VaR^{p}(X_{0}) \quad \mbox{and}
  \end{equation} 
\begin{equation}
\Big(\VaR^{p}_{t}\Big)^{Z}(X)= \displaystyle\sum_{k=1}^{t} (-1)^{t-k}\,(\overline{\mu}_k+\overline{\sigma}_{k}q_{1-p}) +
 (-1)^{t}\,(\mu_0+\sigma_{0}q_{1-p}),
\end{equation}
where
%\begin{equation}
$$
\overline{\mu}_k=<\mu A^{\top},Z_{k}>\;\mbox{and}\;\; \overline{\sigma}_k=<\sigma A^{\top},Z_{k}>.
$$
%\end{equation}

\end{lemme}

\begin{proof}
 The $X_{t}$ are Gaussian and are independent,
and $\EE[Z_{t+1}\mid Z_t]=AZ_t$.
\end{proof}

\begin{lemme}
Suppose that $X_{t}\sim\mathcal{W}(\lambda_{t},\alpha_{t};\theta_{t})$ where 
$\lambda_t\defegal <\lambda,Z_{t+1}> $ % ,\;\alpha_t\defegal <\alpha,Z_{t+1}>$
and $ \theta_t\defegal <\theta,Z_{t+1}>$ with $\lambda,\,\theta\in\RR^N$. % and $\alpha_t\in\RR$.
 Suppose that the dynamic value-at-risk  is driven by the recursion
formula, \emph{i.e.},  
$$\VaR^{p}_{t}(X)=
\displaystyle \sum_{k=1}^{t}(-1)^{t-k}\,\Big(\theta_{k}+\lambda_{k}\big(-\ln (1-p)\big)^{\frac{1}{\alpha_{k}}}\Big) 
+ (-1)^{t} \Big(  \theta_{0}+\lambda_{0}\Big(-\ln (1-p)\Big)^{\frac{1}{\alpha_{0}}} \Big).
$$
Then the Markov modulated dynamic value-at-risk  is
\begin{equation*}
  \Big(\VaR^{p}_{0}\Big)^{Z}(X)\defegal \VaR^{p}(X_{0}) \quad \mbox{and}
  \end{equation*} 
\begin{equation}
\Big(\VaR^{p}_{t}\Big)^{Z}(X)=
\sum_{k=1}^{t}(-1)^{t-k}\,\Big(\overline{\theta}_{k}+\overline{\lambda}_{k}\EE\big[-\big(\ln (1-p)\big)^{\frac{1}{\alpha_{k}}}\mid Z_k\big]\Big) 
+ (-1)^{t} \Big( \theta_{0}+\lambda_{0}\Big(-\ln (1-p)\Big)^{\frac{1}{\alpha_{0}}} \Big),
 \end{equation}
where $\overline{\lambda}_k=<\lambda A^{\top},Z_{k}>$ and $ \overline{\theta}_k=<\theta A^{\top},Z_{k}>$.

\end{lemme}
\begin{proof}
The proof is  immediate.
\end{proof}

\subsection{Conditional Value-at-Risk}%as an alternative to VaR}
Conditional Value-at-Risk has been introduced in the financial risk literature to overcome the limitations of the $\VaR$;
see~\citep{Delbaen1997,Embrechts1999}.
\begin{definition} 
Suppose the random variable $X\in\mathcal{X}$ has a probability density. Then
the $\cvar$ associated is the conditional expectation defined by:
\begin{equation}
\cvar^{p}(X) \defegal
  \EE\big[X \mid X\geq VaR^{p}(X)\big].
\end{equation}
\end{definition}
The following theorem allows us to rewrite the $\cvar$ as a solution of an optimization problem.
Denote  $l_+=\max(l,0)$.  Define
\begin{equation}
 F_{p}(X,\eta) = \EE\big[\eta+\frac{1}{1-p}(X-\eta)_+\big],
\end{equation}
for any random variable $X$.
\begin{theorem}
\label{theocvar}
For any random variable $X$ such that $\EE[X]<+\infty$, the function
$F_{p}(X,\cdot)$ is convex, 
continuously differentiable and 
\begin{equation}
 \cvar^{p}(X)=\min_{\eta\in\RR}F_{p}(X,\eta)\, .
\label{eq:ury}
\end{equation}
The set of solutions of the above optimization problem
\begin{equation}
A_{p}(X)=\displaystyle \argmin_{\eta\in\RR}F_{p}(X,\eta)
\end{equation}
is non-empty, closed and bounded. Then,
\begin{equation*}
VaR^{p}(X)\in \displaystyle
\argmin_{\eta\in\RR}F_{p}(X,\eta)\,\,\mbox{ and }\,\,
\cvar^{p}(X)=F_{p}\big(X,VaR^{p}(X)\big).
\end{equation*}
\label{theo:cvar}
\end{theorem}

\begin{proof} See~\citep{RockUry:risk}.
\end{proof}

With this theorem,  $\cvar$ can be computed without any knowledge of Value-at-Risk.

\begin{lemme}
\label{lemme:dynamiccvar}
Suppose that the random return $(X_{1},\ldots,X_{t})$ are independent and identically distributed.
\begin{enumerate}
\item Suppose that  the return at time $t$, $X_{t} \sim \mathcal{N}(\mu_{t},\sigma_{t})$.
The dynamic conditional value-at risk based on recursion associated with $X$ is:
\begin{equation}
\cvar^{p}_{0}(X)=\cvar^{p}(X_{0}),\quad \mbox{and for}\quad t=1,\ldots,T 
\end{equation}
\begin{equation}
\cvar^{p}_{t}(X)= 
  \left\{
\begin{array}{rl}
& \mu_{t}+\sigma_{t}q_{1-p} -\cvar^{p}_{t-1}(X),  \; \mbox{if} \;  X_{t}\leq \VaR^{p}(L_t)-
2\cvar^{p}_{t-1}(X) ,\, \PP. \;a.s  \\[4mm]
&\mu_{t}-\frac{p}{1-p}\sigma_{t}q_{1-p} +\frac{1+p}{1-p}\cvar^{p}_{t-1}(X), \quad\mbox{elsewhere}. 
\end{array}
\right.
\end{equation}
\item Suppose that  the return at time $t$, $X_{t} \sim \mathcal{W}(\lambda_{t},\alpha_{t};\theta_t)$.
The dynamic conditional value-at risk based on recursion associated with $X$ is:
\begin{equation}
\cvar^{p}_{0}(X)=\cvar^{p}(X_{0}),\quad \mbox{and for}\quad t=1,\ldots,T 
\end{equation}
\begin{equation}
\cvar^{p}_{t}(X)= \displaystyle
  \left\{
\begin{array}{rl}
&  \theta_{t} +\lambda_{t}\big(-\ln(1-p)\big)^{\frac{1}{\alpha_{t}}} - \cvar^{p}_{t-1}(X),                       
  \mbox{if} \, X_{t}\leq  \VaR^{p}(X_{t}) +2\cvar^{p}_{t-1}(X)\,,\\[4mm]%\, \PP. \;a.s \\[4mm]
&\frac{1}{1-p} \EE[X_{t}] -  \frac{p}{1-p}\Big( \theta_{t} +\lambda_{t}\big(-\ln(1-p)\big)^{\frac{1}{\alpha_{t}}} \Big) +  
\frac{1+p}{1-p} \cvar^{p}_{t-1}(X),      \mbox{elsewhere}. 
\end{array}
\right.
\end{equation}
\end{enumerate}
\end{lemme}

\begin{proof}
 We shall prove 1. as the proof of 2. follows the same steps.
 
 From the definition~\ref{defdynamic1}, the dynamic $\cvar^{p}_{t}$ based on recursion 
associated with $X$ is
\begin{equation*}
\cvar^{p}_{t}(X)=\cvar^{p}\big(  X_{t} +\cvar^{p}_{t-1}(X)\big).
\end{equation*}
By the theorem~\ref{theocvar}, the above dynamic $\cvar^{p}_{t}$ can be rewritten as
\begin{equation*}
\cvar^{p}_{t}(X)=\VaR^{p}\big( X_{t}+\cvar^{p}_{t-1}(X)  \big) +\frac{1}{1-p}\EE\Big[  
\Big(X_{t} +\cvar^{p}_{t-1}(X)  -\VaR^{p}\big( X_{t}+\cvar^{p}_{t-1}(X)  \big) \Big)_{+} \Big].
\end{equation*} 
From~\eqref{varnormal_t}, $ \VaR^{p}\big( X_{t}+\cvar^{p}_{t-1}(X)  \big) = \mu_{t} -\cvar^{p}_{t-1}(X) + \sigma_{t} q_{1-p}$.
Then
\begin{eqnarray}
\cvar^{p}_{t}(X)&=& \mu_{t} -\cvar^{p}_{t-1}(X) + \sigma_{t} q_{1-p}+\frac{1}{1-p}\EE\Big[  
\Big(X_{t} +\cvar^{p}_{t-1}(X)  -\mu_{t} +\cvar^{p}_{t-1}(X) - \sigma_{t} q_{1-p}\Big)_{+} \Big],\nonumber\\
&=& \mu_{t} -\cvar^{p}_{t-1}(X) + \sigma_{t} q_{1-p}+\frac{1}{1-p}\EE\Big[  
\Big(X_{t}   -\mu_{t}- \sigma_{t} q_{1-p}+ 2\cvar^{p}_{t-1}(X)\Big)_{+} \Big].\nonumber
\end{eqnarray}
\end{proof}

\begin{lemme} Suppose that  the returns $X_{t}$ are Gaussian and independent where 
$\mu_t\defegal <\mu,Z_{t+1}>$ and $\sigma_t\defegal <\sigma,Z_{t+1}>$  with $\mu,\,\sigma\in\RR^N$.
Suppose the dynamic value-at-risk is driven by the recursion formula. Then the Markov modulated dynamic value-at-risk is
\begin{equation*}
\Big(\cvar^{p}_{0}\Big)^{Z}(X)=\cvar^{p}(X_{0}),\quad \mbox{and for}\quad t=1,\ldots,T 
\end{equation*}
\begin{equation}
\Big(\cvar^{p}_{t}\Big)^{Z}(X)=  \left\{
\begin{array}{rl}
& \overline{\mu}_{t}+\overline{\sigma}_{t}q_{1-p} ,\, \mbox{if} \quad X_{t}\leq \VaR^{p}(X_{t}),\\[4mm]
&\overline{\mu}_{t}+\frac{p}{1-p}\overline{\sigma}_{t}q_{1-p} ,     \quad\mbox{elsewhere}. 
\end{array}
\right.
\end{equation}
\end{lemme}
\begin{proof}
Follow the steps in the proof of Lemma~\ref{lemme:dynamicvar} and use the result of Lemma~\ref{lemme:dynamiccvar}.
\end{proof}

\begin{lemme}
 Suppose that  the returns $X_{t}\sim \mathcal{W}(\lambda_{t},\alpha_{t};\theta_{t})$ and independent where 
$\lambda_t\defegal <\lambda,Z_{t+1}>$ and $\alpha_t\defegal <\alpha,Z_{t+1}>$ 
 with $\lambda,\,\alpha,\,\theta\in\RR^N$.
 Suppose that the dynamic value-at-risk  is driven by the recursion
formula. Then the Markov modulated dynamic conditional value-at-risk  is
\begin{equation*}
\Big(\cvar^{p}_{0}\Big)^{Z}(X)=\cvar^{p}(X_{0}),\quad \mbox{and for}\quad t=1,\ldots,T 
\end{equation*}
\begin{equation}
\Big(\cvar^{p}_t\Big)^{Z}(X)= \displaystyle
  \left\{
\begin{array}{rl}
&  \overline{\theta}_{t} +\overline{\lambda}_{t}\EE\big[\big(-\ln(1-p)\big)^{\frac{1}{\alpha_{t}}}\mid Z_{t}\big] - \big(\cvar^{p}\big)^{Z}_{t-1}(X),                       
  \\[2mm]
  &\hspace{3cm} \mbox{if} \, X_{t}\leq  \big(\VaR^{p}\big)^{Z}(X_{t})+2\big(\cvar^{p}\big)^{Z}_{t-1}(X),\\[4mm]
&\frac{1}{1-p} \EE[X_{t}\mid Z_t] -  \frac{p}{1-p}\Big( \overline{\theta}_{t} +
\overline{\lambda}_{t}\EE\big[\big(-\ln(1-p)\big)^{\frac{1}{\alpha_{t}}}\mid Z_{t}\big]\Big) \\[2mm]
&\hspace{4cm}+\frac{1+p}{1-p} \big(\cvar^{p}_{t-1}\big)^{Z}(X),      \mbox{elsewhere}. 
\end{array}
\right.
\end{equation}

\end{lemme}
\begin{proof}
Follow the steps in the proof of Lemma~\ref{lemme:dynamicvar} and use the result of  Lemma~\ref{lemme:dynamiccvar}.
\end{proof}

\subsection{Illustration through simulation}

We shall illustrate the closed dynamic risk formulas, obtained in the previous Lemmas, in two different problems.
The first, illustrates dynamic risk when the random returns are Gaussian. The parameters of the Gaussian returns are calibrated based 
on the annual MSCI world developed market performance index between 1970 and 2009 and  a Markov Chain. 
The second  focuses on random returns which follow Weibull distributions. The parameters of the Weibull distributions are calibrated based on the historical data of a portfolio index quoted on the NYMEX market exchange, namely the global equity portfolio index. A complete risk/return report on the global equity portfolio index is also available on Bloomberg. The parameters of the Weibull distribution will depend on the state of the economy. Notice that the dynamic risk measures and the Markov modulated risk measures are relatively easy to compute. The results reveal that, the recursive risk measures and the Markov modulated risk measures are bounded above by the static risk measurements and then may lead to a lower capital requirement.
\subsubsection{Gaussian returns}
Assume that  $X_{t} \sim \mathcal{N}(\mu_{t},\sigma_{t})$.  Then $\mu_t\defegal <\mu,Z_{t+1}>\,\mbox{and}\,\sigma_t\defegal <\sigma,Z_{t+1}>$  with $\mu,\,\sigma\in\RR^N$, where $N=2$. To calibrate $\mu_t$ and $\sigma_t$, assume that the return of a  risky asset  $X\sim\mathcal{N}(M,\Sigma)$ (with an initial investment of  \$1~000~US) where the
mean and the standard deviation  are calibrated based on the annual MSCI world 
developed market 
performance index between 1970 and 2009. Hence the
mean is M=\$1~113.3425~US dollars  and the standard deviation is $\Sigma$=\$186.29~US. Then, write 
$\mu\defegal \big[M\!*\!1.05\;\;\;  M\!*\!0.95 \big]$ and
 $\sigma\defegal \big[\Sigma\!*\!1.05\;\;\;  \Sigma\!*\!0.95 \big]$. For the Markov chain, assume that the transition matrix is
\begin{equation*}
A =  \left [
\begin{array}{rl}
0.25 & 0.75 \\[4mm]
0.35 & 0.65
\end{array}
\right ],
\end{equation*}
such that $Z_0= [1\;;\; 0]^{\top}$. The confidence level $p$ is fixed to $0.99$. Consider a horizon of  $10$ days, which is the usual
time horizon in market risk reports. We  then  compute and compare different definitions of market risk: static, dynamic recursive and one scenario of the Markov modulated risk. The results are gathered in Figure~\ref{OutPutNormal}.
\begin{figure}[ht!]
\centering \epsfig{file=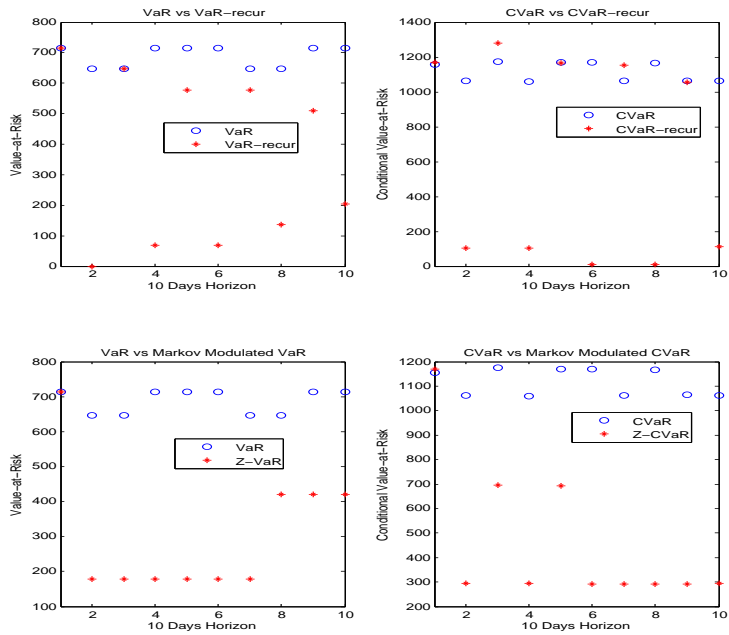,width=12.0cm,height=10.0cm}
\caption{Dynamic Risk  with Gaussian returns.\label{OutPutNormal} }
\end{figure}
Firstly the VaR is below the CVaR, which is consistent of the fact that the VaR is lower than  the CVaR.
Secondly it appears that  recursive and  Markov modulated risk are bounded above by the static risk for both  VaR and CVaR
 most of the time. We have plotted a trajectory of the Markov modulated risk, but more simulations confirm the same behavior. 
Observe that,  recursive risks are more variable than the other risks in the sense that it leads to high risk  on some days and 
low risk on other days. This is because, cash is added in the portfolio to attempt to reduce the risk. 

%\subsection{Market Risk under the hypothesis of  Weibull distribution}
\subsubsection{Returns with Weibull Distribution}
We  begin by estimating the parameters of the Weibull distribution, 
namely $\lambda_t$, $\alpha_t$ and $\theta_t$. To this end, we derive $\lambda_0$ and $\alpha_0$ from real data obtained from
Bloomberg and then generate $\lambda_t$ and $\alpha_t$. We suppose that
$\theta_t=0$, for $t=0,\ldots,T$. The parameters $\lambda$ and $\alpha$ are calibrated based on the returns of the Global Equity Portfolio, quoted $BBGEX$ in Bloomberg, from $12/30/11$ to $09/27/12$. As a traded portfolio index, historical data and 
a complete risk report (returns, risks and risk/return ratios for different horizon times) 
 can be obtained from Bloomberg, see Figure~\ref{BBGEX-Report}. 
 \begin{figure}[b!]
\begin{minipage}{8.0cm}
\centering \epsfig{file=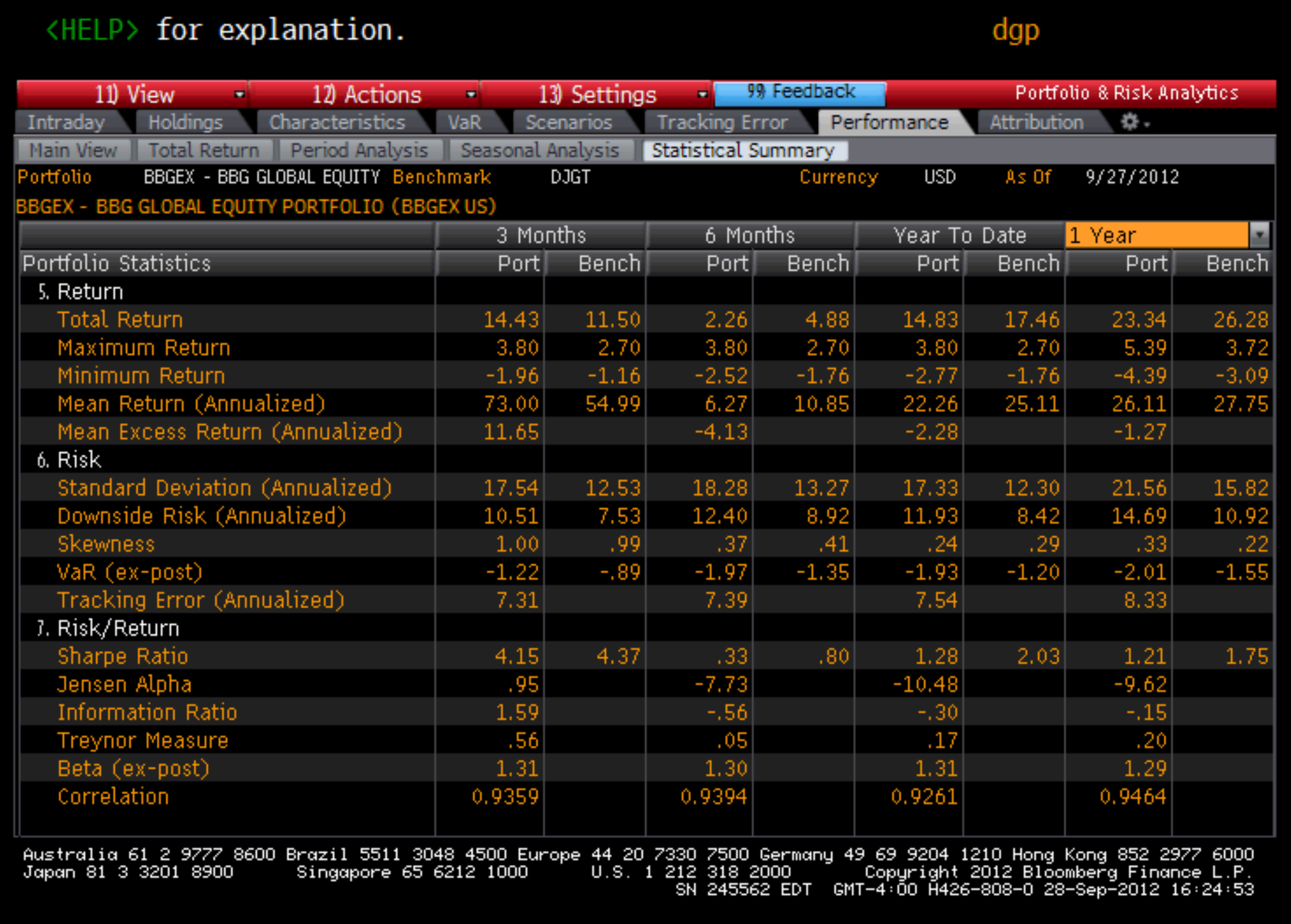,width=8.0cm,height=8.0cm}
\end{minipage}
\begin{minipage}{8.0cm}
\centering \epsfig{file=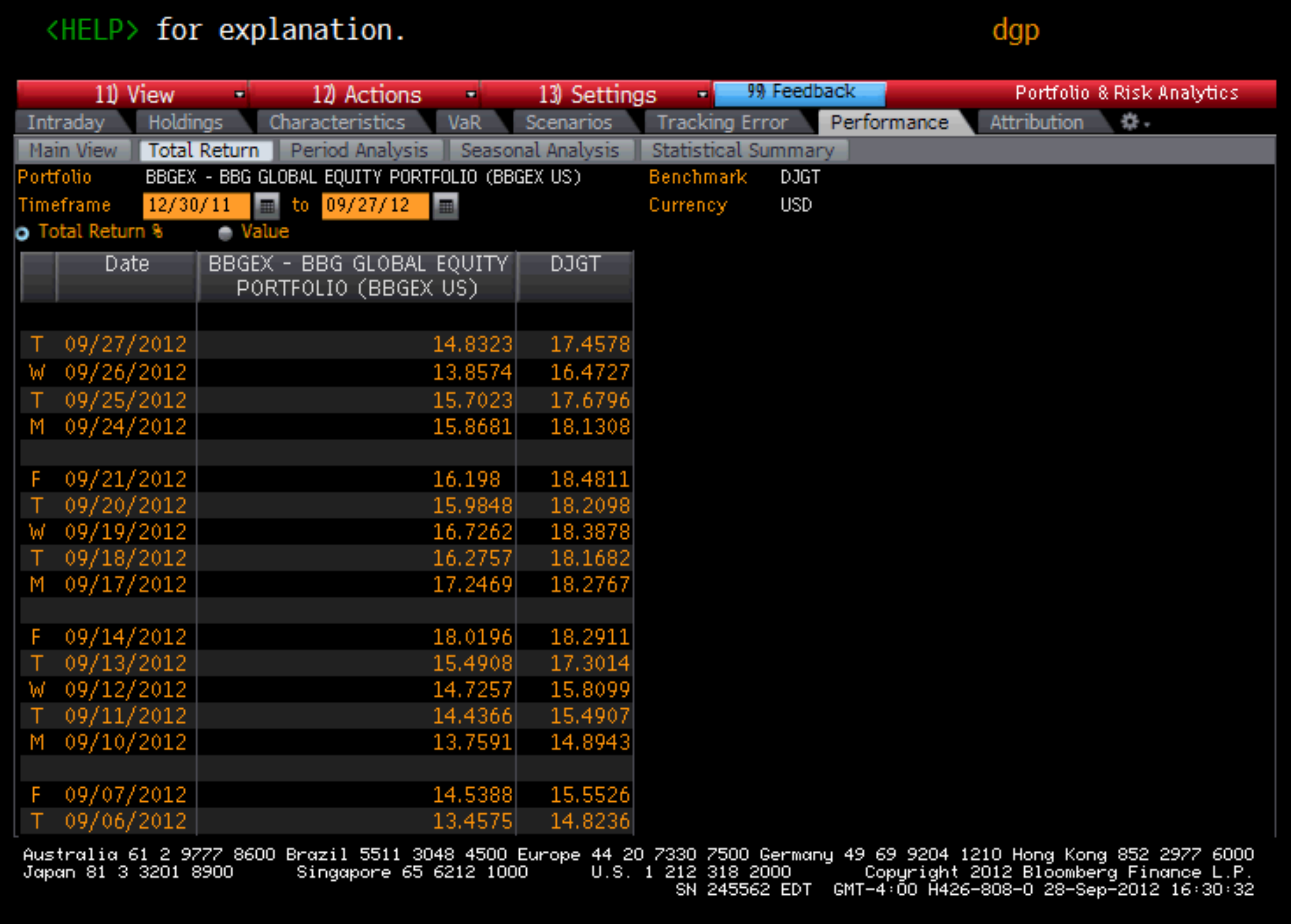,width=7.0cm,height=8.0cm}
\end{minipage}
\caption{Historical Data and Risk Report of the Global Equity Portfolio on  Bloomberg. \label{BBGEX-Report} }
\end{figure}

We then define $\lambda_t\defegal <\lambda,Z_{t+1}>$  with $\lambda\in\RR^2$, and write 
$\lambda\defegal \big[\lambda_0\!*\!1.05\;\;\;  \lambda_0\!*\!0.95 \big]$ and
 $\alpha_t\defegal \alpha_0$ for all $t=1,\ldots T$.
Using the Matlab command, wblfit, the values $\lambda_0=6.7679$ and $\alpha_0= 0.8016$ are obtained. The same values for the matrix of transition $A$ are used. Then different market risks are computed and compared, as shown in Figure~\ref{OutPutWeibull}.
\begin{figure}[t!]
\centering \epsfig{file=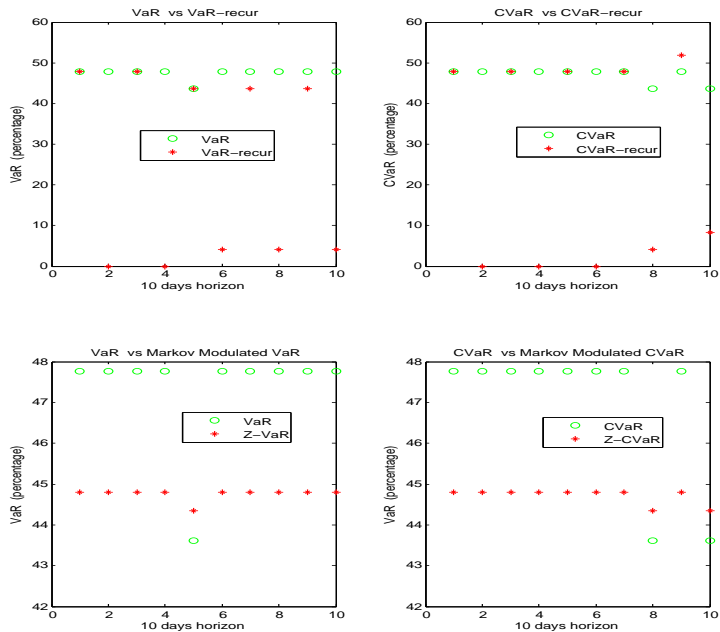,width=12.0cm,height=10.0cm}
\caption{Dynamic Risk  with Weibull returns.\label{OutPutWeibull} }
\end{figure}
The recursive and the Markov modulated risk are again well above the static risk for the VaR and the CVaR.
The behavior of the recursive risk measurements are confirmed, bounded above by the static risk and null sometimes.
The Markov modulated risk are relatively stable and low. 

In conclusion, firstly the dynamic formulation of the market risk measure presented are relatively easy to compute and  interpret, in comparison with the existing formulations in the literature. Secondly, the Markov modulated risk give lower values than the static risk and do not need as much additional cash as required for the risk based on recursion. Also, the Markov modulated risk  may be preferred in terms of risk management as they lead to a lower capital requirement and give, better estimates by taking into account the state of the economy. A drawback is how to calibrate the parameters of the change rate of the economy. Are those parameters dependent on a certain category of returns, or can they be calibrated for any kind of assets? These questions are open and are crucial as the transition matrix  appears in the Markov modulated risk closed formulas.

\section{Conclusion}\label{conclu}
Dynamic risk measure formulas which are easy to compute, are obtained using recursive formulas and state economy representations.
The dynamic risks with recursive formulas  are derived from the invariant translation property.  The Markov modulated risks are based on a finite state discrete Markov chain which represents the state of the economy.  In both cases, the dynamic risk measures inherit properties from the initial risk measure used to express the dynamic risk measures. It appears that, at a given time $t$, the static risk measure  is not necessarily greater than the recursive dynamic risk measures, and the Markov modulated risk gives lower values. It would be interesting to extend the dynamic risk formulas into a continuous time framework.

\bibliographystyle{gENO}

%\appendix

\bibliography{Seck_Biblio}

\end{document}